\documentclass[conference]{IEEEtran}
\IEEEoverridecommandlockouts
\usepackage{
amsmath,cite,bm,amsfonts,amssymb,
graphicx,color,mathtools,setspace,
comment,multirow,xfrac,epstopdf,
footnote,multirow,lscape,float,
caption,subcaption,amsthm,alltt,placeins,amsthm}

\newtheorem{theorem}{Theorem}

\newtheorem*{corollary}{Corollary 1}

\newcommand{\myincludegraphics}{\includegraphics[trim=0cm 0cm 0cm 0.7cm, clip=true, width=0.9\columnwidth]}
\newcommand{\raisecapt}{\vspace{-9pt}}
\newcommand{\myVM}[3]{\mathbf{#1}_{\mathrm{#2}}^{#3}} 
\newcommand{\myVMIndex}[4]{\mathbf{#1}_{\mathrm{#2},#3}^{#4}} 
\newcommand{\Norm}[3]{\left\lVert #1 \right\rVert_{#2}^{#3}} 
\newcommand{\tr}[3]{\mathrm{tr}\left\{\mathbf{#1}_{\mathrm{#2}}^{#3}\right\}}
\newcommand{\eq}[1]{(#1)}
\newcommand{\Diag}[1]{\mathrm{diag}\{#1\}}
\newcommand{\Abs}[2]{\left\lvert #1 \right\rvert^{#2}}
\newcommand{\Exp}[1]{\mathbb{E}\left\{ #1 \right\}}
\newcommand{\Var}[1]{\mathrm{Var}\left\{ #1 \right\}}

\begin{document}

\bstctlcite{IEEEexample:BSTcontrol} 

\title{Optimal SNR Analysis for Single-user RIS Systems}

\author{\IEEEauthorblockN{%
		Ikram Singh\IEEEauthorrefmark{1}, %
		Peter J. Smith\IEEEauthorrefmark{2}, %
		Pawel A. Dmochowski\IEEEauthorrefmark{1}}
	\IEEEauthorblockA{\IEEEauthorrefmark{1}%
		School of Engineering and Computer Science, Victoria University of Wellington, Wellington, New Zealand}
	\IEEEauthorblockA{\IEEEauthorrefmark{2}%
		School of Mathematics and Statistics, Victoria University of Wellington, Wellington, New Zealand}
	\IEEEauthorblockA{email:%
		~\{ikram.singh,peter.smith,pawel.dmochowski\}@ecs.vuw.ac.nz
	}%
}

\maketitle
\begin{abstract}
In this paper, we present an analysis of the optimal uplink SNR of a SIMO Reconfigurable Intelligent Surface (RIS)-aided wireless link. We assume that the channel between base station and RIS is a rank-1 LOS channel while the user (UE)-RIS and UE-base station (BS) channels are correlated Rayleigh. We derive an exact closed form expression for the mean SNR and an approximation for the SNR variance leading to an accurate gamma approximation to the distribution of the UL SNR. Furthermore, we analytically characterise the effects of correlation on SNR, showing that correlation in the UE-BS channel can have negative effects on the mean SNR, while correlation in the UE-RIS channel improves system performance. For systems with a large number of RIS elements, correlation in the UE-RIS channel can cause an increase in the mean SNR of up to 27.32\% relative to an uncorrelated channel. 
\end{abstract}
\IEEEpeerreviewmaketitle
%
%
\section{Introduction}
Reconfigurable Intelligent Surface (RIS) aided wireless networks are currently the subject of considerable research attention due to their ability to manipulate the channel between users (UEs) and base station (BS) via the RIS. Assuming that channel state information (CSI) is known at the RIS, one can intelligently alter the RIS phases, essentially changing the channel, to improve various aspects of system performance. Here, we focus on a single user system and assume that the RIS is carefully located near the BS such that a rank-1 line-of-sight (LOS) channel is formed between the BS and RIS. 

System scenarios with a LOS channel between the BS-RIS and a single-user are also considered in \cite{GMD,Max_Min,OptPHI} with motivation for the LOS assumption given in \cite{Max_Min}. All of these existing works aim to enhance the system to achieve some optimal system performance (e.g. sum rate, SINR, etc.) by tuning the RIS phases. In particular, \cite{OptPHI} and \cite{Max_Min}  provide a closed form RIS phase solution with and without the presence of a direct UE-BS channel for a single user setting, respectively. However, once the optimal RIS has been defined there is no \textit{exact} analysis of the mean SNR and no analysis of correlation impact on the mean SNR in \cite{GMD,Max_Min,OptPHI}.

For the UE to RIS and the direct UE to BS links, scattered channels are a reasonable assumption and spatial correlation in the channels is an important factor, especially at the RIS where small inter-element spacing may be envisaged. Several papers do consider spatial correlation in the small-scale fading channels \cite{Max_Min,Prop_Gauss,Nadeem,Trans_Design,Nadeem2,Two_Timescale}, however, these papers are simulation based and no \textit{exact} analysis is given on the impact of correlation on the mean SNR. 


Statistical properties of the channel have been investigated in existing literature. For example, \cite{GammaApproxY,Relay_Comp} provide a closed form expression for the mean SNR in the absence of a UE-BS channel with \cite{Relay_Comp} additionally providing a probability density function (PDF) and a cumulative distribution function (CDF) for the distribution of the SNR. In \cite{GammaApproxY,Perf_Analy}, an upper bound is given for the ergodic capacity and in \cite{Asymp_Opt} a lower bound is given for the ergodic capacity. However, there is no existing literature on closed form expressions for the mean SNR and SNR variance of an optimal RIS-aided wireless system, in the presence of correlated UE-BS, UE-RIS fading channels. 

Hence, in this paper, we focus on an analysis of the optimal uplink (UL) SNR for a single user RIS aided link with a rank-1 LOS RIS-BS channel and correlated Rayleigh fading for the UE-BS and UE-RIS channels\footnote{An extension to this work is being developed which considers Ricean fading for the UE-BS and UE-RIS channels.}.
In particular, for this system and channel model we make the following contributions: 
\begin{itemize}
	\item An exact closed-form result for mean SNR and an approximate closed form expression for SNR variance are derived. We show that a gamma distribution provides a good approximation of the UL optimal SNR distribution.
	\item Exact closed-form expressions for both mean SNR and SNR variance are derived for uncorrelated Rayleigh channels and presented as a special case.
	\item The analysis is leveraged to gain insight into the impact of spatial correlation and system dimension on the mean SNR. We prove that correlation in the UE-BS channel can have negative effects, while correlation in the UE-RIS channel improves the mean SNR. For systems with a large number of RIS elements, the latter improvement saturates to a relative gain of approximately 27.32\%.
\end{itemize}
\textit{Notation:} $\Exp{\cdot}$ represents statistical expectation. $\Re\left\{ \cdot \right\}$ is the Real operator. $\Norm{\cdot}{2}{}$ denotes the $\ell_{2}$ norm. $\mathcal{CN}(\boldsymbol{\mu},\mathbf{Q})$ denotes a complex Gaussian distribution with mean $\boldsymbol{\mu}$ and covariance matrix $\mathbf{Q}$. $\mathbf{1}_n$ represents an $n \times n$ matrix with unit entries. The transpose, Hermitian transpose and complex conjugate operators are denoted as $(\cdot)^{T},(\cdot)^{H},(\cdot)^{*}$, respectively. The angle of a vector $\myVM{x}{}{}$ of length $N$ is defined as $\angle \myVM{x}{}{} = [ \angle x_{1},\ldots, \angle x_{N} ]^{T}$ and the exponent of a vector is defined as $e^{\myVM{x}{}{}} = [ e^{x_{1}},\ldots,  e^{x_{N}} ]^{T}$.

\section{System Model}\label{Sec: System Model Main}
As shown in Fig.~\ref{Fig: System Model}, we examine a RIS aided single user single input multiple output (SIMO) system where a RIS with $N$ reflective elements is located close to a BS with $M$ antennas such that a rank-1 LOS condition is achieved between the RIS and BS. 
\begin{figure}[h]
	\centering
	\includegraphics[scale=1.1]{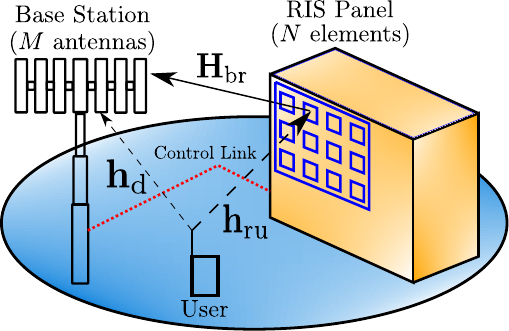}
	\caption{System model.}
	\label{Fig: System Model}
\end{figure}
\subsection{Channel Model}\label{Sec: Channel Model}
Let $\myVM{h}{d}{} \in \mathbb{C}^{M \times 1}$, $\myVM{h}{ru}{} \in \mathbb{C}^{N \times 1}$, $\myVM{H}{br}{} \in \mathbb{C}^{M \times N}$ be the UE-BS, UE-RIS and RIS-BS channels, respectively. The diagonal matrix $\myVM{\Phi}{}{} \in \mathbb{C}^{N \times N}$, where $\mathbf{\Phi}_{r,r} = e^{j\phi_{r}}$ for $r=1,2,\ldots,N$, contains the reflection coefficients for each RIS element. The global UL channel is thus represented by
\begin{equation}\label{Eq: 1}
	\mathbf{h} = \myVM{h}{d}{} + \myVM{H}{br}{} \myVM{\Phi}{}{} \myVM{h}{ru}{}.
\end{equation}
In the channel model, we consider the correlated Rayleigh channels: $\myVM{h}{d}{} = \sqrt{\beta_{\textrm{d}}} \myVM{R}{d}{1/2} \myVM{u}{d}{}$ and $\myVM{h}{ru}{} = \sqrt{\beta_{\textrm{ru}}} \myVM{R}{ru}{1/2} \myVM{u}{ru}{}\triangleq \sqrt{\beta_{\textrm{ru}}}\myVM{\tilde{h}}{ru}{} $ where $\beta_{\textrm{d}}$ and $\beta_{\textrm{ru}}$ are the link gains, $\myVM{R}{d}{}$ and $\myVM{R}{ru}{}$ are the correlation matrices for UE-BS and UE-RIS links respectively and $\myVM{u}{d}{},\myVM{u}{ru}{} \sim \mathcal{CN}(\mathbf{0},\mathbf{I})$. The rank-1 LOS channel from RIS to BS has link gain $\beta_{\textrm{br}}$ and is given by $\myVM{H}{br}{} = \sqrt{\beta_{\textrm{br}}}\myVM{a}{b}{}\myVM{a}{r}{H}$ where $\myVM{a}{b}{}$ and $\myVM{a}{r}{}$ are topology specific steering vectors at the BS and RIS respectively. Particular examples of steering vectors for a vertical uniform rectangular array (VURA) are in Sec.~\ref{Sec: Results}.


Note, that the correlation matrices $\myVM{R}{ru}{}$ and $\myVM{R}{d}{}$ can represent any correlation model. For simulation purposes, we will use the well-known exponential decay model,
\begin{equation}
	(\myVM{R}{ru}{})_{ik}= \rho_{\mathrm{ru}}^{\frac{d_{i,k}}{d_{\textrm{b}}}} 
	\quad , \quad
	(\myVM{R}{d}{})_{ik} = \rho_{\mathrm{d}}^{\frac{d_{i,k}}{d_{\textrm{r}}}},
\end{equation}
where $0 \leq \rho_{\text{ru}} \leq 1$, $0 \leq \rho_{\text{d}} \leq 1$. $d_{i,k}$ is the distance between the $i^{\textrm{th}}$ and $k^{\textrm{th}}$ antenna/element at the BS/RIS. $d_{\textrm{b}}$ and $d_{\textrm{r}}$ are the nearest-neighbour BS antenna separation and RIS element separation, respectively, which are measured in wavelength units. $\rho_{\mathrm{d}}$ and $\rho_{\mathrm{ru}}$ are the nearest neighbour BS antenna and RIS element correlations, respectively.

\subsection{Optimal SNR}\label{Sec: Optimal SNR}
Using  (\ref{Eq: 1}), the received signal at the BS is,
$
\mathbf{r} = \left(\myVM{h}{d}{} + \myVM{H}{br}{} \myVM{\Phi}{}{} \myVM{h}{ru}{} \right)s + \mathbf{n} \triangleq \myVM{h}{}{}s + \mathbf{n},
$
where $s$ is the transmitted signal with power $E_{s}$ and $\mathbf{n} \sim \mathcal{CN}(\mathbf{0},\sigma^{2}\mathbf{I})$. For a single user system, matched filtering (MF) is the optimal combining method, with UL SNR, given by
$
\text{SNR} = \myVM{{h}}{}{H}\myVM{{h}}{}{}\bar{\tau},
$
where $\bar{\tau} = \frac{E_{s}}{\sigma^{2}}$. The optimal RIS phase matrix to maximize the SNR can be computed using the main steps outlined in \cite[Sec. III-B]{OptPHI} but using an UL channel model instead of downlink. Substituting the channel vectors and through some algebraic manipulation, the optimal RIS phase matrix is
\begin{equation}\label{Eq: Optimum PHI}
	\myVM{\Phi}{}{} = 
	\frac{\myVM{a}{b}{H} \myVM{h}{d}{}}{\Abs{\myVM{a}{b}{H} \myVM{h}{d}{}}{}}
	\Diag{e^{j\angle\myVM{a}{r}{}}}
	\Diag{e^{-j\angle\myVM{h}{ru}{}}}.
\end{equation}
Substituting \eq{\ref{Eq: Optimum PHI}} into $\myVM{h}{}{}$, the optimal UL global channel is,
\vspace{-5pt}
\begin{equation}\label{Eq: Final r signal}
	\myVM{h}{}{} =  \myVM{h}{d}{} + \sqrt{\beta_{\mathrm{br}}\beta_{\mathrm{ru}}} \psi \sum_{n=1}^{N}\Abs{\myVMIndex{\tilde{h}}{ru}{n}{}}{} \myVM{a}{b}{} ,
\end{equation}
where $\psi = \myVM{a}{b}{H} \myVM{h}{d}{}/|\myVM{a}{b}{H} \myVM{h}{d}{}|$. Hence, the optimal UL SNR is
\begin{equation}\label{Eq: SNR Eq}
	\text{SNR} = 
	\left( \myVM{h}{d}{H} + \alpha^{*}\myVM{a}{b}{H} \right)
	\left( \myVM{h}{d}{} + \alpha^{}\myVM{a}{b}{} \right)
	\bar{\tau},
\end{equation}
where $\alpha = \sqrt{\beta_{\textrm{br}}\beta_{\textrm{ru}}} \psi \sum_{n=1}^{N}\Abs{\myVMIndex{\tilde{h}}{ru}{n}{}}{} \triangleq  \sqrt{\beta_{\textrm{br}}\beta_{\textrm{ru}}} \psi Y$.

\section{$\mathbb{E}\{\text{SNR}\}$ and $\textsc{var}\{\text{SNR}\}$}\label{Sec: E{SNR} and var{SNR}}

\subsection{Correlated Rayleigh Case}\label{SubSec: Corr Rayl}
Here, we provide an exact closed form result for $\mathbb{E}\{\text{SNR}\}$ and an exact expression for $\Var{\text{SNR}}$.
\begin{theorem}
The mean SNR is given by
\begin{align}\label{Eq: Corr E{SNR}}
	& \Exp{\mathrm{SNR}} = \notag \\
	& \left( 
	\hspace{-0.3em}\beta_{\mathrm{d}}M  +
	\dfrac{NA\pi\sqrt{\beta_{\mathrm{d}}\beta_{\mathrm{br}}\beta_{\mathrm{ru}}} }{2}  
	+ \beta_{\mathrm{br}}\beta_{\mathrm{ru}}M(N + F) 
	\right)\bar{\tau},
\end{align}
with $F$ given by
\begin{equation}\label{Eq: HyperGeometric}
	F = \underset{i \neq j}{\sum_{i=1}^{N} \sum_{j=1}^{N}} \dfrac{\pi}{4}\left( 1 - \left\lvert\rho_{ij}\right\rvert^2\right)^2 {}_{2}F_{1}\left(\frac{3}{2},\frac{3}{2};1;\left\lvert\rho_{ij}\right\rvert^2 \right),
\end{equation}
and $A = \left\lVert \myVM{R}{d}{1/2} \myVM{a}{b}{} \right\rVert _{2}$, where ${}_{2}F_{1}(\cdot)$ is the Gaussian hypergeometric function and $\rho_{ij} = \left(\myVM{R}{ru}{} \right)_{ij}$. 
\end{theorem}
\begin{proof}
See App.~\ref{AppA: E{SNR} Corr Deri} for the derivation of \eq{\ref{Eq: Corr E{SNR}}}.
\end{proof}
\begin{theorem}
The SNR variance is given by
\begin{align}\label{Eq: Corr var{SNR}}
\Var{\mathrm{SNR}}&=\Big(\beta_{\mathrm{d}}^{2} \tr{R}{d}{2} 
	+ \beta_{\mathrm{d}}^{3/2}\sqrt{\beta_{\mathrm{br}}\beta_{\mathrm{ru}}}N\pi(B - MA) \notag \\
	& + \beta_{\mathrm{d}}\beta_{\mathrm{br}}\beta_{\mathrm{ru}} A^{2} \left(4(N+F) - \dfrac{N^2\pi^2}{4} \right)\notag \\
	& + M A \sqrt{\beta_{\mathrm{d}}}\left(\beta_{\mathrm{br}}\beta_{\mathrm{ru}}\right)^{3/2} \left( 2\sqrt{\pi}C_{1} - N(N+F)\pi\right) 	\notag \\
	& + (M\beta_{\mathrm{br}}\beta_{\mathrm{ru}})^2 \left(C_{2} - (N+F)^2\right) \Big) \bar{\tau}^{2} ,
\end{align}
where $ B = MA +{\myVM{a}{b}{H} \myVM{R}{d}{2} \myVM{a}{b}{}}/{2A} $, $F$ is given by (\ref{Eq: HyperGeometric}), $A$ is given in Theorem 1 and $C_{1}=\Exp{Y^{3}},C_{2}=\Exp{Y^{4}}$. 
\end{theorem}
\begin{proof}See App.~\ref{AppB: var{SNR} Corr Deri} for the derivation of \eq{\ref{Eq: Corr var{SNR}}}. \end{proof}
In \eq{\ref{Eq: Corr var{SNR}}}, all terms are known except for $C_{1}$ and $C_{2}$, the third and fourth moments of $Y$. To the best of our knowledge, these moments are intractable so we present approximations of $C_{1}$ and $C_{2}$ in the following Corollary.
\begin{corollary}
	To obtain an SNR variance approximation in closed form, we approximate the 3$^{\text{rd}}$, 4$^{\text{th}}$ moments of $Y$ by,
	\begin{align}\label{Eq: Corollary}
		\hspace{ -5pt}\Exp{Y^3}  & =  b^3 a \prod_{k=1}^{2}(k+a), \quad
		\Exp{Y^4}  =  b^4 a \prod_{k=1}^{3}(k+a)
	\end{align}
with 
\begin{equation*}
	a = \dfrac{N^{2}\pi}{4(N+F) - N^{2}\pi} \quad,\quad
	b = \dfrac{2}{N\sqrt{\pi}}\left( N + F - \dfrac{N^{2}\pi}{4} \right).
\end{equation*}
where $F$ is given by \eq{\ref{Eq: HyperGeometric}}.
\end{corollary}
\begin{proof}See App.~\ref{AppD: E{Y^3},E{Y^4} Approx} for the derivation of \eq{\ref{Eq: Corollary}}. \end{proof}
From Theorem 1 and the Corollary, we have $\Exp{\text{SNR}}$ and an approximation to $\Var{\text{SNR}}$. This enable us to fit a gamma distribution as an approximation to the SNR distribution. Motivation and reasoning for this approach is given in Sec.~\ref{SubSec: Gamma approx fit for SNR}.
\subsection{Special Case: Uncorrelated Rayleigh Case}\label{SubSec: UnCorr Rayl}
For independent Rayleigh fading, we provide exact closed form expressions for both $\mathbb{E}\{\text{SNR}\}$ and $\text{Var}\{\text{SNR}\}$. Here, $\myVM{h}{d}{},\myVM{h}{ru}{} \sim \mathcal{CN}(\mathbf{0},\mathbf{I})$ and as such, (\ref{Eq: HyperGeometric}) simplifies to
\begin{equation}\label{Eq: Unc HyperGeometric}
F_{\mathrm{u}} = \underset{i \neq j}{\sum_{i=1}^{N} \sum_{j=1}^{N}} \dfrac{\pi}{4} = N(N-1)\dfrac{\pi}{4}.
\end{equation}
From $\Norm{\myVM{a}{b}{}}{2}{} = \sqrt{M}$, the value of $\Exp{\text{SNR}}$ is,
\begin{align}\label{Eq: Unc E{SNR}}
\hspace{-0.3em} \left( \hspace{-0.3em} \beta_{\mathrm{d}}M + 
\dfrac{\sqrt{M}N\pi}{2} \sqrt{\beta_{\mathrm{d}}\beta_{\mathrm{br}}\beta_{\mathrm{ru}}} 
 + \beta_{\mathrm{br}}\beta_{\mathrm{ru}}M(N + F_{\mathrm{u}}) \hspace{-0.3em} \right) \hspace{-0.3em} \bar{\tau}.
\end{align}

\noindent Using \eq{\ref{Eq: Unc HyperGeometric}}, an exact expression for $\text{Var}\left\{\text{SNR}\right\}$ is given by
\begin{align}\label{Eq: Unc var{SNR}}
 &\Big(\beta_{\mathrm{d}}^{2} M 
+ \beta_{\mathrm{d}}^{3/2}\sqrt{\beta_{\mathrm{br}}\beta_{\mathrm{ru}}}N\pi(B_{\mathrm{u}} - 2M^{3/2}) \notag \\
& + \beta_{\mathrm{d}}\beta_{\mathrm{br}}\beta_{\mathrm{ru}} M \left(4(N+F_{\mathrm{u}}) - \dfrac{N^2\pi^2}{4} \right) \notag\\
& + M^{3/2} \sqrt{\beta_{\mathrm{d}}}\left(\beta_{\mathrm{br}}\beta_{\mathrm{ru}}\right)^{3/2} \left( 2\sqrt{\pi}C_{\mathrm{u1}} - N(N+F_{\mathrm{u}})\pi \right) 	\notag \\
& + (M\beta_{\mathrm{br}}\beta_{\mathrm{ru}})^2 \left(C_{\mathrm{u2}} - (N+F_{\mathrm{u}})^2\right) \Big) \bar{\tau}^{2},
\end{align}
with
\begin{equation*}
	B_{\mathrm{u}} = M^{3/2} + \dfrac{\sqrt{M}}{2},
	\hspace{3pt}
	C_{\mathrm{u1}} = \dfrac{N\sqrt{\pi}}{2}\hspace{-2pt}\left(\dfrac{\pi}{4}\prod_{k=1}^{2}(N-k) 
	\hspace{-2pt}+\hspace{-2pt} 3N \hspace{-2pt}-\hspace{-2pt} \frac{3}{2} \right)
\end{equation*}
\begin{align*}
	C_{\mathrm{u2}} & = 2N + \binom{N}{2}\left(\prod_{k=2}^{3}(N-k)\dfrac{\pi^{2}}{8} + 6 + 3\pi(N-1)\right)
\end{align*}
where  $F_{\mathrm{u}}$ is given by (\ref{Eq: Unc HyperGeometric}). Derivations for $C_{\mathrm{u1}}$, $C_{\mathrm{u2}}$ are omitted for reasons of space, but can be easily obtained by expanding $Y^{3}$, $Y^{4}$ and computing the expectation of each term using the following: $\Exp{\Abs{\myVMIndex{\tilde{h}}{ru}{i}{}}{4}}=2$, $\Exp{\Abs{\myVMIndex{\tilde{h}}{ru}{i}{}}{3}}=3\sqrt{\pi}/4$, $\Exp{\Abs{\myVMIndex{\tilde{h}}{ru}{i}{}}{2}}=1$, $\Exp{\Abs{\myVMIndex{\tilde{h}}{ru}{i}{}}{}}=\sqrt{\pi}/2$.


\section{Performance insights based on $\Exp{\text{SNR}}$}\label{Sec: Asymp Anly}
Correlation in $\myVM{h}{d}{}$ and $\myVM{h}{ru}{}$ have separable effects on the mean SNR in \eq{\ref{Eq: Corr E{SNR}}}. Specifically, the second term in \eq{\ref{Eq: Corr E{SNR}}} is only affected by $\rho_{\textrm{d}}$ and the third term is only affected by $\rho_{\textrm{ru}}$ whilst the first term is affected by neither. Here, we present an analysis of $\Exp{\text{SNR}}$ with respect to the correlations $\rho_{\mathrm{ru}},\rho_{\mathrm{d}}$ and give asymptotic results for large $N$.

\subsection{Effect of $\rho_{\mathrm{ru}}$ on $\Exp{\text{SNR}}$}\label{SubSec: rho_ru vs E{SNR{}}}
With respect to the correlation level $\rho_{\mathrm{ru}}$, \eq{\ref{Eq: Corr E{SNR}}} can be lower and upper bounded due to $F$ being monotonic in $\rho_{\mathrm{ru}}$ \cite[Eq. (15.1.1)]{Stegan}.
The upper bound for \eq{\ref{Eq: HyperGeometric}} as $\rho_{\mathrm{ru}} \rightarrow 1$ is
\begin{equation}\label{Eq: Hypergeometric Func UB}
	F_{\mathrm{UB}} 
	\overset{(a)}{=} 
	\underset{i \neq j}{\sum_{i=1}^{N} \sum_{j=1}^{N}} \frac{\pi}{4}
	{}_{2}F_{1}\left(-\frac{1}{2},-\frac{1}{2};1;1 \right)
	\overset{(b)}{=} 
	N(N-1),
\end{equation} 
where $(a)$ uses \cite[Eq. (15.3.3)]{Stegan} to perform a linear transformation of the hypergeometric function and $(b)$ uses \cite[Eq. (15.1.20)]{Stegan} to reduce the hypergeometric function to a known value along with evaluation of the double summation. Using \eq{\ref{Eq: Unc HyperGeometric}} and \eq{\ref{Eq: Hypergeometric Func UB}}, \eq{\ref{Eq: HyperGeometric}} has the following upper and lower bounds,
\begin{equation}\label{Eq: HyperGeometric func LB/UB}
	F_{\mathrm{u}}=N(N-1)\frac{\pi}{4}
	\leq
	F 
	\leq
	N(N-1).
\end{equation}

\noindent Hence the upper and lower bounds for SNR w.r.t. $\rho_{\text{ru}}$ are
\begin{equation}\label{Eq: Corr E{SNR} UB}
	\mathbb{E}\left\{\text{SNR}\right\}_{\mathrm{UB}} \hspace{-0.3em}=\hspace{-0.3em} 
	\left( \beta_{\mathrm{d}}M + 
	\dfrac{NA\pi}{2} \sqrt{\beta_{\mathrm{d}}\beta_{\mathrm{br}}\beta_{\mathrm{ru}}} + \beta_{\mathrm{br}}\beta_{\mathrm{ru}}MN^{2} \right)
	\hspace{-0.3em}\bar{\tau},
\end{equation}
\begin{align}\label{Eq: Corr E{SNR} LB}
& \mathbb{E}\left\{\text{SNR}\right\}_{\mathrm{LB}} = \notag \\
& \left( \beta_{\mathrm{d}}M + 
\dfrac{NA\pi}{2} \sqrt{\beta_{\mathrm{d}}\beta_{\mathrm{br}}\beta_{\mathrm{ru}}} 
+ \beta_{\mathrm{br}}\beta_{\mathrm{ru}}M(N + F_{\mathrm{u}}) \right) \bar{\tau}.
\end{align}
We observe that correlation in $\myVM{h}{ru}{}$ improves the mean SNR. From \eqref{Eq: Corr E{SNR} UB} and \eqref{Eq: Corr E{SNR} LB} we observe that the first term in $\mathbb{E}\left\{\text{SNR}\right\}$ is of order $M$ and the third term is of order $MN^2$. In Sec.~\ref{SubSec: rho_d vs E{SNR{}}} we show that the second term has a maximum order of $MN$. Hence, the third term is dominant giving $\mathbb{E}\left\{\text{SNR}\right\}=\mathcal{O}(MN^2)$ and the largest channel effect will be an increase in SNR due to correlation at the RIS.

\subsection{Effect of $\rho_{\mathrm{d}}$ on $\Exp{\text{SNR}}$}\label{SubSec: rho_d vs E{SNR{}}}
The only variable in \eq{\ref{Eq: Corr E{SNR}}} that is affected by $\rho_{\text{d}}$ is $A$. As $\rho_{\mathrm{d}} \rightarrow 0$, $A  \rightarrow \sqrt{M}$. As $\rho_{\mathrm{d}} \rightarrow 1$, $\myVM{R}{d}{} \rightarrow \mathbf{1}_{M}$ which means that 
\begin{equation*}
	A 
	= \Norm{\myVM{R}{d}{1/2}\myVM{a}{b}{}}{2}{}
	 \rightarrow \frac{1}{\sqrt{M}}\Norm{\mathbf{1}_{M}\myVM{a}{b}{}}{2}{} 
	= \Abs{\sum_{i=1}^{M} \myVMIndex{a}{b}{i}{}}{} \leq M.
\end{equation*}
Hence, the second term in \eq{\ref{Eq: Corr E{SNR}}} has a maximum order of $MN$ as stated above. Although the maximum order can be achieved with perfect correlation ($\rho_{\text{d}}=1$),    for typical environments the value of $A$ tends to reduce as correlation is increased. To explain this property we show via an example based on a uniform linear array (ULA) that for highly correlated BS antennas, $\sqrt{M}$ tends to be larger than $\Abs{\sum_{i=1}^{M} \myVMIndex{a}{b}{i}{}}{}$ for large $M$. For a ULA, ignoring elevation angles, the steering vector elements can be given by $\myVMIndex{a}{b}{i}{} = e^{j 2 \pi (i-1) d_{\mathrm{b}} \sin(\theta)}$ where $\theta$ is the azimuth angle of arrival at the BS. Here,
\begin{equation*}
	\Abs{\sum_{i=1}^{M} \myVMIndex{a}{b}{i}{}}{}
	=  \Abs{\frac{1-e^{j2M\pi d_{\mathrm{b}} \sin(\theta)}}{1-e^{j2\pi  d_{\mathrm{b}} \sin(\theta)}}}{}
	= \Abs{\frac{\sin(M\pi d_{\mathrm{b}} \sin(\theta))}{\sin(\pi d_{\mathrm{b}} \sin(\theta))}}{}.
\end{equation*}
This well-known sinusoidal ratio is much smaller than $\sqrt{M}$ for large $M$ and as long as $\theta$ is not arbitrarily close to zero. Therefore, systems with a large number of BS antennas can expect greater system performance when $\myVM{h}{d}{}$ is uncorrelated unless the RIS-BS link is extremely close to broadside. 

In summary, high correlation at the RIS and low correlation at the BS tend to be beneficial. We refer to this scenario as the \textit{favorable channel scenario}.

\subsection{Favorable Channel Scenario and Asymptotic Analysis}\label{SubSec: Asymp Anly ideal channel}
Using the analysis in Sec. \ref{SubSec: rho_ru vs E{SNR{}}} and Sec. \ref{SubSec: rho_d vs E{SNR{}}}, the favorable channel scenario is given by: $\myVM{h}{ru}{} \sim \mathcal{CN}(\mathbf{0},\mathbf{1}_{N})$, $\myVM{h}{d}{} \sim \mathcal{CN}(\mathbf{0},\mathbf{I}_{M})$. The resulting mean SNR is obtained by substituting the UB of \eq{\ref{Eq: HyperGeometric func LB/UB}} and $A=\sqrt{M}$ into \eq{\ref{Eq: Corr E{SNR}}}, giving
\begin{align}\label{Eq: Ideal E{SNR}}
	& \mathbb{E}\left\{\text{SNR}_{\text{fav}}\right\} \notag \\
	& = \left( 
	\beta_{\mathrm{d}}M + 
	\dfrac{N\sqrt{M}\pi}{2} \sqrt{\beta_{\mathrm{d}}\beta_{\mathrm{br}}\beta_{\mathrm{ru}}} +
	 \beta_{\mathrm{br}}\beta_{\mathrm{ru}}MN^{2} \right)\bar{\tau}.
\end{align}

Next, we consider the asymptotic gains achievable through increased correlation at the RIS. The relative gain due to correlation can be defined as
\begin{align}\label{Eq: Avg SNR Gain}
	\text{Gain}_{\textrm{corr}}
	& = \frac{\Exp{\text{SNR}}_{\text{UB}} - \Exp{\text{SNR}}_{\text{LB}}}{\Exp{\text{SNR}}_{\text{LB}}} \notag \\
	& \overset{(a)}{=} 
	\frac{(4-\pi)N^{2} + (\pi-4)N}{\pi N^{2} + (4-\pi)N + \frac{4\beta_{\mathrm{d}}}{\beta_{\mathrm{br}}\beta_{\mathrm{ru}}} + \frac{2NA\pi\sqrt{\beta_{\mathrm{d}}}}{M\sqrt{(\beta_{\mathrm{br}}\beta_{\mathrm{ru}})}}}
\end{align}
where $(a)$ involves substituting \eq{\ref{Eq: Corr E{SNR} UB}} and \eq{\ref{Eq: Corr E{SNR} LB}} and performing simple algebraic manipulations. Therefore, as $N \rightarrow \infty$,
\begin{equation*}
	\text{Gain}_{\text{max}} 
	= \lim\limits_{N \rightarrow \infty} \text{Gain}_{\textrm{corr}}
	= \frac{4-\pi}{\pi} \approx 27.32\%.
\end{equation*}
Hence, for a large RIS, the maximum gain due to correlation in the UE-RIS Rayleigh channel is approximately 27.32\%. Also, observe that negative effects on the mean SNR due to correlation in the direct channel are minimized as $N$ increases.
\section{Results}\label{Sec: Results}
We present numerical results to verify the analysis in Secs. \ref{Sec: E{SNR} and var{SNR}} and \ref{Sec: Asymp Anly}. We do not consider cell-wide averaging as the focus is on the SNR distribution over the fast fading. Furthermore, the relationship between the SNR and the gains, $\beta_{\mathrm{d}},\beta_{\mathrm{br}}$ and $\beta_{\mathrm{ru}}$, is straightforward, as shown in \eq{\ref{Eq: Corr E{SNR}}}. Hence, we present numerical results for fixed link gains. In particular, as the RIS-BS link is LOS we assume $\beta_{\mathrm{br}}=d_{\mathrm{br}}^{-2}$ where $d_{\mathrm{br}}=20$m. Next, for simplicity, $\beta_{\mathrm{d}}=\beta_{\mathrm{ru}}=0.59$. This was chosen to give the 95\%-ile of the SNR distribution as 25 dB in the baseline case of moderate channel correlation (see $\rho_{\mathrm{d}}=\rho_{\mathrm{ru}}=0.7$ in Fig.~\ref{Fig: CDF Agreements}), with $M=32,N=64$.


As stated in Sec. \ref{Sec: Channel Model}, the steering vectors for $\myVM{H}{br}{}$ are not restricted to any particular formation. However, for simulation purposes, we will use the VURA model as outlined in \cite{CMiller}, but in the $y-z$ plane with equal spacing in both dimensions at both the RIS and BS. The $y,z$ components of the steering vector at the BS are $\myVM{a}{b,y}{} $ and $\myVM{a}{b,z}{}$ which are given by
\begin{align*}
	 & [1, e^{j2\pi d_{\mathrm{b}} \sin(\theta_{\mathrm{A}})\sin(\omega_{\mathrm{A}})}, \ldots, e^{j2\pi d_{\mathrm{b}} (M_{y}-1) \sin(\theta_{\mathrm{A}})\sin(\omega_{\mathrm{A}})}]^{T}, \\
	 & [1, e^{j2\pi d_{\mathrm{b}} \cos(\theta_{\mathrm{A}})}, \ldots, e^{j2\pi d_{\mathrm{b}}(M_{z}-1) \cos(\theta_{\mathrm{A}})}]^{T},
\end{align*}
respectively. Similarly at the RIS, $\myVM{a}{r,y}{}$ and $\myVM{a}{r,z}{}$ are defined by,
\begin{align*}
	& [1, e^{j2\pi d_{\mathrm{r}} \sin(\theta_{\mathrm{D}})\sin(\omega_{\mathrm{D}})}, \ldots, e^{j2\pi d_{\mathrm{r}}(N_{y}-1) \sin(\theta_{\mathrm{D}})\sin(\omega_{\mathrm{D}})}]^{T}, \\
	& [1, e^{j2\pi d_{\mathrm{r}} \cos(\theta_{\mathrm{D}})}, \ldots, e^{j2\pi d_{\mathrm{r}} (N_{z}-1) \cos(\theta_{\mathrm{D}})}]^{T},
\end{align*}
respectively where $M = M_{y}M_{z}$, $N = N_{y}N_{z}$, $d_{\mathrm{b}}=0.5$, $d_{\mathrm{r}}=0.2$, where $d_{\mathrm{b}}$ and $d_{\mathrm{r}}$ are in wavelength units. Therefore, the steering vectors at the BS and RIS are then given by,
\begin{equation}
	\myVM{a}{b}{} = \myVM{a}{b,y}{} \otimes \myVM{a}{b,z}{}
	\quad,\quad
	\myVM{a}{r}{} = \myVM{a}{r,y}{} \otimes \myVM{a}{r,z}{},
\end{equation} 
respectively, where  $\otimes$ denotes the Kronecker product, $\theta_{\mathrm{A}}$ and $\omega_{\mathrm{A}}$ are elevation/azimuth angles of arrival (AOAs) at the BS and $\theta_{\mathrm{D}},\omega_{\mathrm{D}}$ are the corresponding angles of departure (AODs) at the RIS. The elevation/azimuth angles are selected based on the following geometry representing a range of LOS $\myVM{H}{br}{}$ links with less elevation variation than azimuth variation: 
$	
\theta_{D} \sim \mathcal{U}[70^{o},90^{o}], \hspace{1em}
\omega_{D} \sim \mathcal{U}[-30^{o},30^{o}], \hspace{1em}  
\theta_{A} = 180^{o} - \theta_{D}, \hspace{1em} 
\omega_{A} \sim \mathcal{U}[-30^{o},30^{o}] 
$ where $\mathcal{U}[a,b]$ denotes a uniform random variable taking on values between $a$ and $b$. For all results in this paper we use a single sample from this range of angles given by $\theta_{D}=77.1^{o}, \omega_{D}=19.95^{o}, \theta_{A}=109.9^{o}, \omega_{A}=-29.9^{o}$.
Note that all of these parameter values and variable definitions are not altered throughout the results and figures, unless specified otherwise.
\subsection{Approximate CDF for SNR}\label{SubSec: Gamma approx fit for SNR}
It is known that the SNR of a wide range of fading channels can be approximated by a mixture gamma distribution \cite{GammaApproxSNR}. Also, it is well-known that a single gamma approximation is often reasonable for a sum of a number of positive random variables \cite{GammaApproxY}. Motivated by this, we approximate the SNR in \eq{\ref{Eq: SNR Eq}} by a single gamma variable.

The shape parameter of a gamma approximation to the SNR is given by $k_{\gamma}=\frac{\Exp{\text{SNR}}^{2}}{\Var{\text{SNR}}}$ and the scale parameter is $\theta_{\gamma}=\frac{\Var{\text{SNR}}}{\Exp{\text{SNR}}}$ where $\Exp{\text{SNR}}$ and $\Var{\text{SNR}}$ are given in Sec. \ref{Sec: E{SNR} and var{SNR}}. Using these values of $k_{\gamma},\theta_{\gamma}$, the analytical and simulated SNR CDFs  are shown in Fig. \ref{Fig: CDF Agreements} for $N=64$ and $N=256$, both with  $\rho_{\mathrm{ru}}=\rho_{\mathrm{d}} \in \{0,0.7,0.95\}$. When computing the analytical SNR CDFs, for $\rho_{\mathrm{ru}}=\rho_{\mathrm{d}}=0$, \eq{\ref{Eq: Unc E{SNR}}} and \eq{\ref{Eq: Unc var{SNR}}} were used and for $\rho_{\mathrm{ru}}=\rho_{\mathrm{d}}\neq0$, \eq{\ref{Eq: Corr E{SNR}}} and \eq{\ref{Eq: Corr var{SNR}}} were used.
\begin{figure}[h]
	\centering
	\myincludegraphics{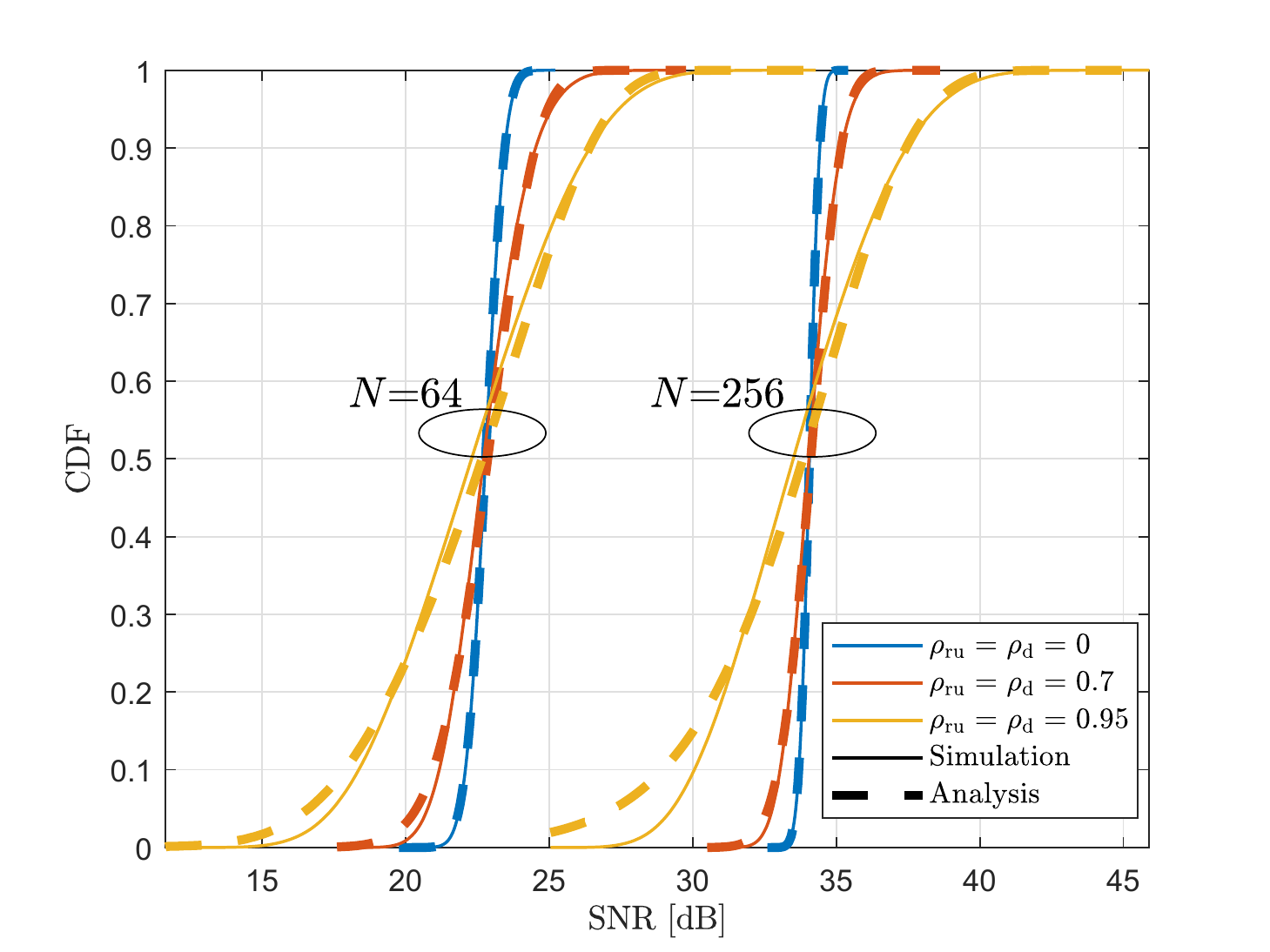}
	\raisecapt\caption{Simulated and analytical CDFs for $N=64$ and $N=256$, both with  $\rho_{\mathrm{ru}}=\rho_{\mathrm{d}}=\{0,0.7,0.95\}$}
	\label{Fig: CDF Agreements}
\end{figure}
As expected, there is a very good agreement between the simulated and analytical SNR CDFs when $\rho_{\mathrm{d}}=\rho_{\mathrm{d}}=0$ due to exact mean SNR and SNR variance expressions. Increasing the correlation level causes the CDF agreement to deviate slightly in the low SNR region, especially in the highest correlation scenario. Good agreement is maintained in the mid-high SNR region. The gamma distribution therefore provides a good representation of the UL SNR unless the correlations become very high.

\subsection{$\rho_{\mathrm{d}},\rho_{\mathrm{ru}}$ and Asymptotic Analysis Results}\label{SubSec: Asymptotic Analysis Results}
Here, we verify the performance insights based on $\Exp{\text{SNR}}$. Fig. \ref{Fig: Rho analysis} represents SNR simulations and analysis for three different correlation  scenarios: $\rho_{\text{ru}}=\rho_{\text{d}}=0$, $\rho_{\text{ru}}=\rho_{\text{d}}=1$ and the favorable channel scenario $\rho_{\text{ru}}=1,\rho_{\text{d}}=0$.
\begin{figure}[h]
	\centering
	\myincludegraphics{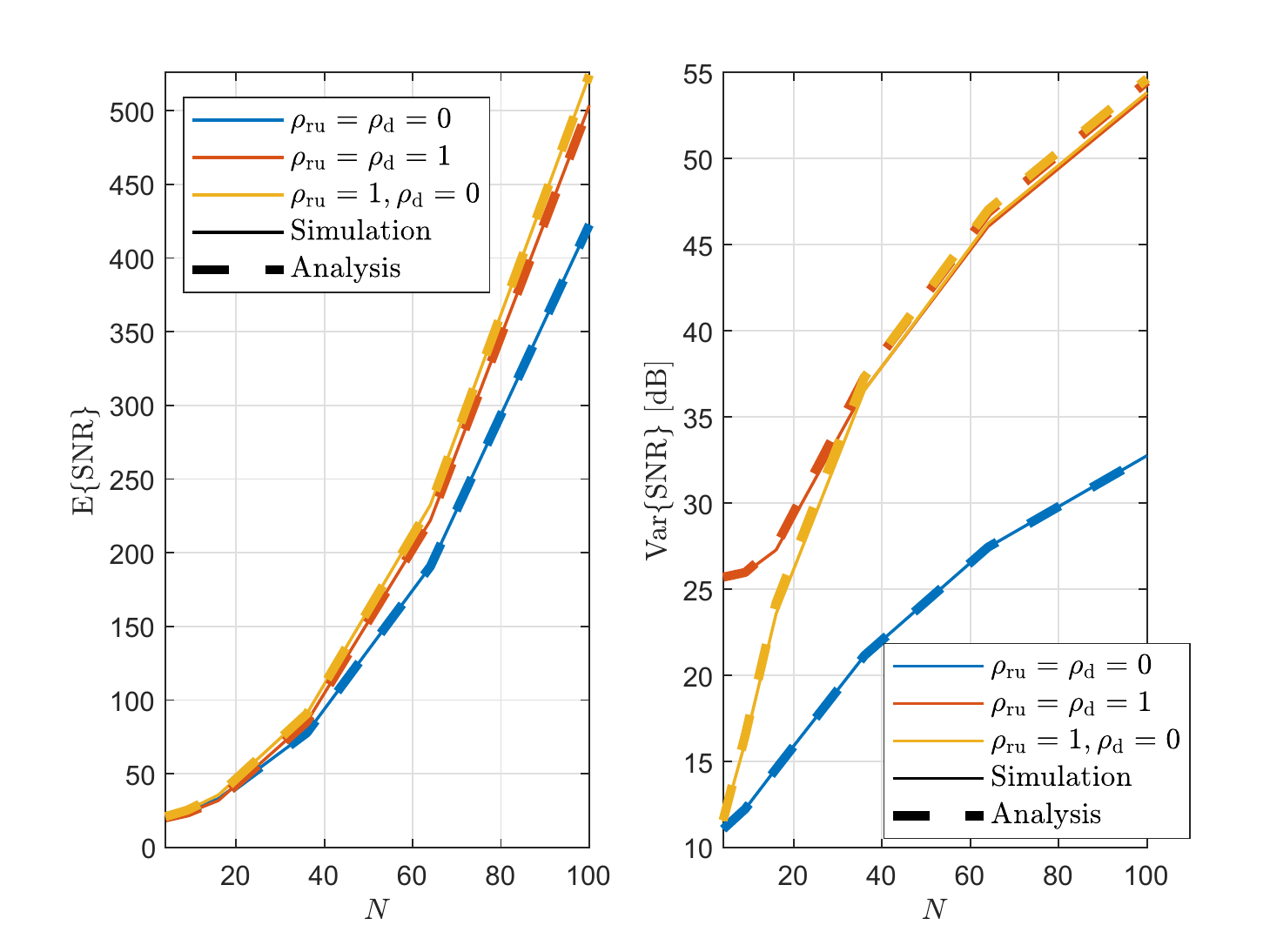}
	\raisecapt\caption{Simulated and analytical mean and variance results for SNR with three different correlation scenarios: $\rho_{\text{ru}}=\rho_{\text{d}}=0$, $\rho_{\text{ru}}=\rho_{\text{d}}=1$ and favorable channel scenario $\rho_{\text{ru}}=1,\rho_{\text{d}}=0$.}
	\label{Fig: Rho analysis}
\end{figure}
The left subfigure in Fig. \ref{Fig: Rho analysis} shows the quadratic increase in mean SNR, as predicted by the analysis. The favorable channel scenario yields the highest mean SNR, but the increase over perfect correlation in both channels is marginal. The lowest SNR occurs when both channels are uncorrelated. Fig. \ref{Fig: Rho analysis} also shows that for all correlation scenarios, the theoretical analysis agrees with simulations.

The right subfigure in Fig. \ref{Fig: Rho analysis} shows the accuracy of the SNR variance approximation. There is a perfect agreement for the case where $\myVM{h}{ru}{}$ and $\myVM{h}{d}{}$ are uncorrelated. In the correlated cases, the analysis and simulation agree closely for low $N$ but begin to deviate slightly as $N$ grows. This is also reflected in Fig. \ref{Fig: CDF Agreements} as the CDF agreement between simulation and analysis is worse for $N=256$ compared to $N=64$ as $\rho_{\text{d}},\rho_{\text{ru}} \rightarrow 1$. Note that the analytical CDFs show a longer lower tail, partially caused by the over-estimate of the variance. As $N$ grows, observe that the variance for scenarios $\rho_{\mathrm{ru}}=\rho_{\mathrm{d}}=1$ and  $\rho_{\mathrm{ru}}=1,\rho_{\mathrm{d}}=0$ converge to approximately the same value since the effects of correlation in $\myVM{h}{d}{}$ are reduced by large $N$. 

Finally, in Fig. \ref{Fig: Asymptotic Analysis} we verify the mean SNR relative gain due to correlation in $\myVM{h}{ru}{}$, and verify the asymptotic analysis in Sec. \ref{SubSec: rho_ru vs E{SNR{}}} and Sec. \ref{SubSec: Asymp Anly ideal channel}. For simplicity, we assume all three channels have the same link gain and let $\beta_{\mathrm{d}}=\beta_{\mathrm{ru}}=\beta_{\mathrm{br}}=1$.
\begin{figure}[h]
	\centering
	\myincludegraphics{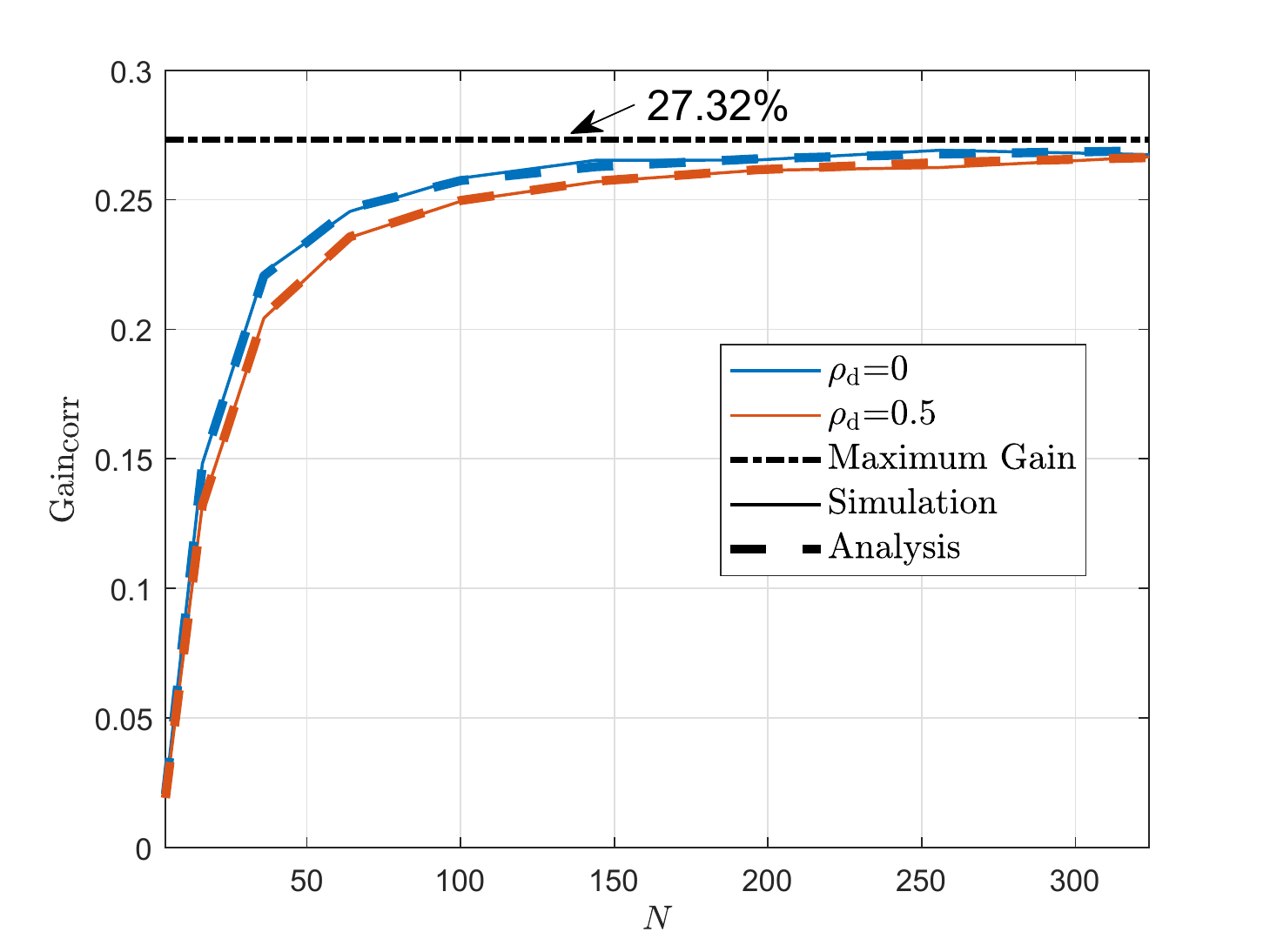}
	\raisecapt\caption{Average SNR gain due to correlation in $\myVM{h}{ru}{}$ for varying RIS sizes and correlation in $\myVM{h}{d}{}$. }
	\label{Fig: Asymptotic Analysis}
\end{figure}
Fig. \ref{Fig: Asymptotic Analysis} verifies the analysis in Sec. \ref{SubSec: Asymp Anly ideal channel} demonstrating an increasing gain with correlation saturating at approximately 27.32\%. Introducing correlation in the direct channel causes the gain to be lower. However, as explained in the analysis, for large RIS elements, this negative effect is reduced.

\section{Conclusion}\label{Sec: Conclusion}
 We derive an exact closed form expression for the mean SNR  of the optimal single user RIS design where spatially correlated Rayleigh fading is assumed for the UE-BS and UE-RIS channels and the RIS-BS channel is LOS. We also provide an accurate approximation to the SNR variance and a gamma approximation to the CDF of the SNR. The results offers new insight into how spatial correlation impacts the mean SNR and scenarios in which we would expect high SNR performance.

\vspace{-12pt}
\begin{appendices} 
\section{$\mathbb{E}\{\text{SNR}\}$ derivation in Sec.~\ref{SubSec: Corr Rayl}}\label{AppA: E{SNR} Corr Deri}
Substituting the channel vectors and matrices described in Sec.~\ref{Sec: Channel Model} into \eq{\ref{Eq: SNR Eq}} gives
\begin{align*}
	\text{SNR} &= 
	\left( \beta_{\mathrm{d}}\myVM{u}{d}{H}\myVM{R}{d}{}\myVM{u}{d}{} + 
	2 \sqrt{\beta_{\mathrm{d}}}\Re\left\{ \alpha \myVM{u}{d}{H}\myVM{R}{d}{1/2}\myVM{a}{b}{} \right\} + 
	\Abs{\alpha}{2}M \right) \bar{\tau}\notag \\
& \triangleq (S_1+S_2+S_3)\bar{\tau}.
\end{align*}
We compute $\mathbb{E}\left\{ \text{SNR} \right\}$ by taking the expectation of each term.

\noindent \textbf{Term 1}: Since $\myVM{u}{d}{} \sim \mathcal{CN}(\mathbf{0},\mathbf{I})$,
\begin{equation}\label{AppA: part 1}
	\Exp{S_1} = \beta_{\mathrm{d}} \tr{R}{d}{} = \beta_{\mathrm{d}} M.
\end{equation}

\noindent \textbf{Term 2}:
Substituting $\alpha$ from Sec.~\ref{Sec: Optimal SNR} we have,
\begin{align*}
	\Exp{S_2}
	& = 2\sqrt{\beta_{\mathrm{d}}\beta_{\mathrm{br}}\beta_{\mathrm{ru}}} 
	\mathbb{E}\left\{ Y \right\} 
	\mathbb{E}\left\{ \Abs{\myVM{a}{b}{H} \myVM{R}{d}{1/2} \myVM{u}{d}{}}{} \right\}.
\end{align*}
Since $ \myVM{h}{ru}{} \sim \mathcal{CN}(\mathbf{0},\myVM{R}{ru}{})$ it follows that $\mathbb{E}\left\{ Y \right\} = \frac{N\sqrt{\pi}}{2}$ \cite{8567939}. To compute $\mathbb{E}\left\{ \Abs{\myVM{a}{b}{H} \myVM{R}{d}{1/2} \myVM{u}{d}{}}{} \right\}$, note the following results. As $\myVM{u}{d}{} \sim \mathcal{CN}(\boldsymbol{0},\mathbf{I}_{M})$, it follows that $\myVM{a}{b}{H} \myVM{R}{d}{1/2} \myVM{u}{d}{}$ is zero mean complex Gaussian with
$
	\Exp{|\myVM{a}{b}{H} \myVM{R}{d}{1/2} \myVM{u}{d}{}|^2  }
	= \Norm{\myVM{R}{d}{1/2}\myVM{a}{b}{}}{2}{2}.
$
As such, we have $\Abs{\myVM{a}{b}{H} \myVM{R}{d}{1/2} \myVM{u}{d}{}}{} \sim 
\frac{1}{\sqrt{2}}\Norm{\myVM{R}{d}{1/2}\myVM{a}{b}{}}{2}{} X^{1/2}$ where $X \sim \chi_{2}^{2}$ is a central chi-squared variable. By the moments of a central chi-square distribution with 2 degrees of freedom,
\begin{equation}
	\mathbb{E}\left\{\Abs{\myVM{a}{b}{H} \myVM{R}{d}{1/2} \myVM{u}{d}{}}{} \right\} = \Norm{\myVM{R}{d}{1/2}\myVM{a}{b}{}}{2}{} \dfrac{\sqrt{\pi}}{2}.
\end{equation}
Hence,
\begin{equation}\label{AppA: part 2}
	\Exp{S_2} = \dfrac{NA\pi}{2} \sqrt{\beta_{\mathrm{d}}\beta_{\mathrm{br}}\beta_{\mathrm{ru}}},
\end{equation}
where $A = \Norm{\myVM{R}{d}{1/2}\myVM{a}{b}{}}{2}{}$.

\noindent \textbf{Term 3}:
Using  $\Abs{\psi}{} = 1$ (where $\psi$ is given in Sec.~\ref{Sec: Optimal SNR}) we have $S_3= 
M\beta_{\mathrm{br}}\beta_{\mathrm{ru}} Y^{2}$ and expanding $Y$ gives
\begin{align*}
	\mathbb{E}\left\{ Y^{2} \right\} & = 
	\sum_{i=1}^{N}\mathbb{E}\left\{ \Abs{\myVMIndex{\tilde{h}}{ru}{i}{}}{2} \right\}
	+ \underset{i \neq j}{\sum_{i=1}^{N} \sum_{j=1}^{N}} \mathbb{E}\left\{ \Abs{\myVMIndex{\tilde{h}}{ru}{i}{}}{} \Abs{\myVMIndex{\tilde{h}}{ru}{j}{}}{}  \right\}.
\end{align*}
Using \cite[Eq. (16)]{8567939}, each term in the double summation is,
$$
\mathbb{E}\left\{ \Abs{\myVMIndex{\tilde{h}}{ru}{i}{}}{} \Abs{\myVMIndex{\tilde{h}}{ru}{j}{}}{}  \right\} 
= \dfrac{\pi}{4}\left( 1 - \left\lvert\rho_{ij}\right\rvert^2\right)^2 {}_{2}F_{1}\left(\frac{3}{2},\frac{3}{2};1;\left\lvert\rho_{ij}\right\rvert^2 \right),
$$
where ${}_{2}F_{1}(\cdot)$ is the Gaussian hypergeometric function and $\rho_{ij} = \left(\myVM{R}{ru}{} \right)_{ij}$. Using this, we have $ \mathbb{E}\left\{ Y^{2} \right\} = N + F, $
where $F$ is given by \eq{\ref{Eq: HyperGeometric}}, giving the final result
\begin{equation}\label{AppA: part 3}
	\Exp{S_3} = \beta_{\mathrm{br}}\beta_{\mathrm{ru}}M(N+F).
\end{equation}
Combining \eq{\ref{AppA: part 1}}, \eq{\ref{AppA: part 2}} and \eq{\ref{AppA: part 3}} completes the derivation.

\section{$\textsc{var}\{\text{SNR}\}$ derivation in Sec.~\ref{SubSec: Corr Rayl}}\label{AppB: var{SNR} Corr Deri}
To compute the variance we take the square of \eq{\ref{Eq: SNR Eq}} giving,
\begin{align}\label{Eq: SNR^2}
	\text{SNR}^{2} & = 
	\Big( \beta_{\mathrm{d}}^{2} \left( \myVM{u}{d}{H}\myVM{R}{d}{}\myVM{u}{d}{} \right)^{2} + 
	 4\beta_{\mathrm{d}}^{3/2} \myVM{u}{d}{H}\myVM{R}{d}{}\myVM{u}{d}{} \Re\left\{\alpha \myVM{u}{d}{H}\myVM{R}{d}{1/2}\myVM{a}{b}{}  \right\} \notag \\
	& +  2\beta_{\mathrm{d}} M \Abs{\alpha}{2} \myVM{u}{d}{H}\myVM{R}{d}{}\myVM{u}{d}{}  + 
	4\beta_{\mathrm{d}}\Re\left\{\alpha \myVM{u}{d}{H}\myVM{R}{d}{1/2}\myVM{a}{b}{}  \right\}^{2} \notag \\
	& + 4 \sqrt{\beta_{\mathrm{d}}} M \Abs{\alpha}{2} \Re\left\{ \alpha\myVM{u}{d}{H}\myVM{R}{d}{1/2}\myVM{a}{b}{} \right\} + \Abs{\alpha}{4} M^{2}  \Big) \bar{\tau}^{2}\notag \\
& \triangleq (T_1+T_2+T_3+T_4+T_5+T_6)\bar{\tau}^{2}.
\end{align}
Terms 1, 3, 4, 5 and 6 can be computed using standard results, the results in \cite[Eq. (9)]{8422818} and the methods given in App.~\ref{AppA: E{SNR} Corr Deri}. The results are,
\begin{equation}\label{AppB: Part 1}
	\Exp{T_1} = \beta_{\mathrm{d}}^{2} \left( \tr{R}{d}{2} + M^{2} \right),
\end{equation}
\begin{equation}\label{AppB: Part 3}
	\Exp{T_3} = 2\beta_{\mathrm{d}}\beta_{\mathrm{br}}\beta_{\mathrm{ru}} M^{2} (N+F),
\end{equation}
\begin{equation}\label{AppB: Part 4}
	\Exp{T_4}= 4\beta_{\mathrm{d}}\beta_{\mathrm{br}}\beta_{\mathrm{ru}}(N+F)\Norm{\myVM{R}{d}{1/2}\myVM{a}{b}{}}{2}{2},
\end{equation}
\begin{equation}\label{AppB: Part 5}
	\Exp{T_5} \hspace{-0.2em}=\hspace{-0.2em}2 M \sqrt{\pi} \sqrt{\beta_{\mathrm{d}}}\left( \beta_{\mathrm{br}}\beta_{\mathrm{ru}} \right)^{3/2} \Norm{\myVM{R}{d}{1/2}\myVM{a}{b}{}}{2}{} \hspace{-0.3em} \Exp{Y^3},
\end{equation}
\begin{equation}\label{AppB: Part 6}
	\Exp{T_6} = \left( M\beta_{\mathrm{br}}\beta_{\mathrm{ru}} \right)^{2} \Exp{Y^4}.
\end{equation}
The variables $Y^3$ and $Y^4$ are the sum of products of magnitudes of 3 and 4 correlated complex Gaussian random variables. To the best of our knowledge the mean of such terms is intractable without the use of multiple infinite summations and special functions. As such we use an approximation based on the gamma distribution to approximate $\Exp{Y^3},\Exp{Y^4}$ (see App.~\ref{AppD: E{Y^3},E{Y^4} Approx}). Term 2 requires more work and is derived below. 

\noindent \textbf{Term 2}: Expanding the second term gives,
\begin{align}\label{AppB: Part 2}
	\Exp{T_2}= 4 \beta_{\mathrm{d}}^{3/2}\sqrt{\beta_{\mathrm{br}}\beta_{\mathrm{ru}}} Y \myVM{u}{d}{H}\myVM{R}{d}{}\myVM{u}{d}{} 
	\Abs{\myVM{u}{d}{H}\myVM{R}{d}{1/2}\myVM{a}{b}{}}{},
\end{align}
where $Y$ is defined in Sec.~\ref{Sec: Optimal SNR}. To find the expectation of \eq{\ref{AppB: Part 2}}, we introduce the following variables: Let $\mathbf{P}$ be any orthonormal matrix with first column equal to $\mathbf{p}_{1} = {\myVM{R}{d}{1/2}\myVM{a}{b}{}}\Norm{\myVM{R}{d}{1/2}\myVM{a}{b}{}}{2}{-1}$. Also let $\myVM{x}{}{}=\mathbf{P}^{H}\myVM{u}{d}{} \sim \mathcal{CN}(\mathbf{0},\mathbf{I})$ and $\mathbf{Q}=\myVM{P}{}{H}\myVM{R}{d}{}\myVM{P}{}{}$, then the random component of \eq{\ref{AppB: Part 2}}, neglecting $Y$, can be rewritten as,
\begin{align*}
	\myVM{u}{d}{H}\myVM{R}{d}{}\myVM{u}{d}{} 
	\Abs{\myVM{u}{d}{H}\myVM{R}{d}{1/2}\myVM{a}{b}{}}{}  
	& = \myVM{x}{}{H}\myVM{Q}{}{}\myVM{x}{}{} \Abs{x_{1}}{}
	\Norm{\myVM{R}{d}{1/2}\myVM{a}{b}{}}{2}{},
\end{align*}
since $\myVM{P}{}{H}\myVM{p}{1}{} = [1,\mathbf{0}_{M-1}^{T}]^{T}$. Note that,
\begin{align*}
	\Exp{\myVM{x}{}{H}\myVM{Q}{}{}\myVM{x}{}{}\Abs{x_{1}}{}} 
	& 
	= \Exp{ \sum_{i=1}^{M}\sum_{j=1}^{M} \mathbf{Q}_{ij}x_{i}^{*}x_{j}\Abs{x_{1}}{} } \\
	& = \mathbf{Q}_{11}\Exp{\Abs{x_{1}}{3}} + \left(\tr{Q}{}{}-\myVM{Q}{11}{}\right)\Exp{\Abs{x_{1}}{}} \\
	& = \dfrac{\sqrt{\pi} \myVM{a}{b}{H}\myVM{R}{d}{2}\myVM{a}{b}{}}{4 A^2} + M\dfrac{\sqrt{\pi}}{2} 
\end{align*}
since $\Exp{|x_1|}=\sqrt{\pi}/2$, $\Exp{|x_1|^3}=3\sqrt{\pi}/4$, $\tr{Q}{}{} = \tr{R}{d}{} = M$ and $\mathbf{Q}_{11} = \myVM{p}{1}{H}\myVM{R}{d}{}\myVM{p}{1}{} = \frac{\myVM{a}{b}{H}\myVM{R}{d}{2}\myVM{a}{b}{}}{\myVM{a}{b}{H}\myVM{R}{d}{}\myVM{a}{b}{}}$ and $A$ is defined in App.~\ref{AppA: E{SNR} Corr Deri}. Using the above result and the result for $\Exp{Y}$ in App.~\ref{AppA: E{SNR} Corr Deri}, the expectation of \eq{\ref{AppB: Part 2}} is
\begin{align}\label{AppB: part 2 Final}
	& \Exp{T_2} = \beta_{\mathrm{d}}^{3/2}\sqrt{\beta_{\mathrm{br}}\beta_{\mathrm{ru}}}N B \pi,
\end{align}
where $B = MA + {\myVM{a}{b}{H} \myVM{R}{d}{2} \myVM{a}{b}{}}/{2A}$.

Combining \eq{\ref{AppB: Part 1}}, \eq{\ref{AppB: part 2 Final}},  \eq{\ref{AppB: Part 3}}, \eq{\ref{AppB: Part 4}}, \eq{\ref{AppB: Part 5}}, \eq{\ref{AppB: Part 6}} for $\Exp{\text{SNR}^2}$  and subtracting $\Exp{\text{SNR}}^2$ completes the derivation.

\section{Approximations for $\Exp{Y^3}$ and $\Exp{Y^4}$}\label{AppD: E{Y^3},E{Y^4} Approx}
Due to $Y$ being positive, unimodal and the sum of $N$ variables, we propose approximations for $\Exp{Y^3}$ and $\Exp{Y^4}$ using a gamma distribution as an approximation for $Y$ (using the same motivation as in Sec.~\ref{SubSec: Gamma approx fit for SNR}).
From App.~\ref{AppA: E{SNR} Corr Deri}, we know that $\Exp{Y}={N\sqrt{\pi}}/{2}$ and  $\Exp{Y^2} = N+F$, where $F$ is defined by \eq{\ref{Eq: HyperGeometric}}. Then, the variance of $Y$ is $\Var{Y}= N + F - \frac{N^{2}\pi}{4}$. Using the method of moments, the shape and scale parameters that define the gamma fit for $Y$ are, 
\begin{align*}
a = \dfrac{N^{2}\pi}{4(N+F) - N^{2}\pi}, \quad
b = \dfrac{2}{N\sqrt{\pi}}\left( N + F - \dfrac{N^{2}\pi}{4} \right),
\end{align*}
where $a$ and $b$ are the shape and scale parameters respectively. The 3$^{\text{rd}}$ and $4^{\text{th}}$ moments of $Y$ are approximated by, 
\begin{align*}
\Exp{Y^3}  & =  b^3 a \prod_{k=1}^{2}(k+a), \quad
\Exp{Y^4}  =  b^4 a \prod_{k=1}^{3}(k+a).
\end{align*}
Thus we have the results for $C_{1},C_{2}$ in Sec. \ref{SubSec: Corr Rayl}.

\end{appendices}

\bibliographystyle{IEEEtran}
\bibliography{RIS_Paper_ICC}

\end{document}